\documentclass[12pt,a4paper]{article}
\usepackage[UKenglish]{babel}
\usepackage{amsfonts,amsmath,amsthm,amssymb}
\usepackage[pdftex]{graphicx}
\usepackage{clrscode3e}
\usepackage{fancybox}
\usepackage[colorinlistoftodos]{todonotes}

\usepackage[margin=1.2in]{geometry}
\newtheorem{thm}{Theorem}

\newtheorem{lemma}[thm]{Lemma}

\newcommand{\RR}{\ensuremath{\mathbb{R}}}

\renewcommand{\S}{\ensuremath{\mathcal{S}}}
\DeclareMathOperator{\dist}{dist}
\DeclareMathOperator{\cell}{cell}
\DeclareMathOperator{\polylog}{polylog}

\begin{document}

\title{Shortest Paths in Intersection Graphs of Unit Disks}
\author{Sergio Cabello\thanks{Department of Mathematics, IMFM, and Department of Mathematics, FMF, University of Ljubljana, Slovenia. Supported by the Slovenian Research Agency, program P1-0297, projects J1-4106 and L7-5459, and by the ESF EuroGIGA project (project GReGAS) of the European Science Foundation. E-mail: \texttt{sergio.cabello@fmf.uni-lj.si}.}\and Miha Jej\v{c}i\v{c}\thanks{Faculty of Mathematics and Physics, University of Ljubljana, Slovenia. E-mail: \texttt{jejcicm@gmail.com}.}}

\maketitle

\begin{abstract}
	Let $G$ be a unit disk graph in the plane def\mbox{}ined by $n$ disks 
    whose positions are known.
	For the case when $G$ is unweighted, 
	we give a simple algorithm to compute a shortest path tree from a given source 
    in ${\mathcal O}(n\log n)$ time. 
    For the case when $G$ is weighted, 
	we show that a shortest path tree from a given source 
	can be computed in ${\mathcal O}(n^{1+\varepsilon})$ time, improving
	the previous best time bound of ${\mathcal O}(n^{4/3+\varepsilon})$.
\end{abstract}

\section{Introduction}

Each set $\S$ of geometric objects in the plane def\mbox{}ines its intersection graph in a natural way: the vertex set is $\S$ and there is an edge $ss'$ in the graph, $s,s'\in S$, whenever $s\cap s'\not= \emptyset$. It is natural to seek  faster algorithms when the input is constraint to geometric intersection graphs. Here we are interested in computing shortest path distances in unit disk graphs, that is, the intersection graph of equal sized disks. 

A unit disk graph is uniquely def\mbox{}ined by the centers of the disks. Thus, we will drop the use of disks and just refer to the graph $G(P)$ def\mbox{}ined by a set $P$ of $n$ points in the plane. The vertex set of $G(P)$ is $P$. Each edge of $G(P)$ connects points $p$ and $p'$ from $P$ whenever $\|p-p'\|\le 1$, where $\|\cdot\|$ denotes the Euclidean norm. See Figure~\ref{fig:GP} for an example of such graph. Up to a scaling factor, $G(P)$ is isomorphic to a unit disk graph. In the \emph{unweighted} case, each edge $pp'\in E(G(P))$ has unit weight, while in the \emph{weighted} case, the weight of each edge $pp'\in E(G(P))$ is $\|p-p'\|$. In all our algorithms we assume that $P$ is known. Thus, the input is $P$, as opposed to the abstract graph $G(P)$. 

\begin{figure}[htb]
  \centering
  \includegraphics{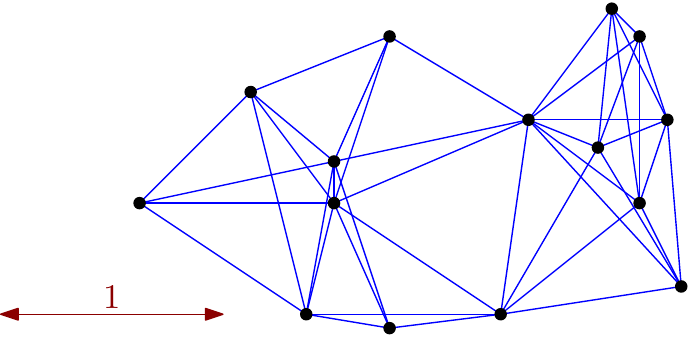}
  \caption{Example of graph $G(P)$.}
  \label{fig:GP}
\end{figure}

Exact computation of shortest paths in unit disks is considered by Roditty and Segal~\cite{rs-11}, under the name of \emph{bounded leg} shortest path problem. They show that, for the weighted case,
a shortest path tree can be computed in ${\mathcal O}(n^{4/3+\varepsilon})$ time.
They also note that the dynamic data structure for nearest neighbors of Chan~\cite{chan-10} imply that, in the unweighted case, shortest paths can be computed in ${\mathcal O}(n \log^6 n)$ expected time. 
(Roditty and Segal~\cite{rs-11} also consider data structures to $(1+\varepsilon)$-approximate shortest path distances in the intersection graph of congruent disks when the size of the disks is given at query time; they improve previous bounds of Bose et al.~\cite{bmnsz-04}. In this paper we do not consider that problem.)

Alon Efrat pointed out that a semi-dynamic data structure described by Efrat, Itai and Katz~\cite{eik-01} can be used to compute in ${\mathcal O}(n\log n)$ time a shortest path tree in the unweighted case. Given a set of $n$ unit disks in the plane, they construct in ${\mathcal O}(n\log n)$ time a data structure that, in ${\mathcal O}(\log n)$ amortized time, f\mbox{}inds a disk containing a query point and deletes it from the set. By repetitively querying this data structure, one can build a shortest path tree from any given source in ${\mathcal O}(n\log n)$ time in a straightforward way. At a very high level, the idea of the data structure is to consider a regular grid of constant-size cells and, for each cell of the grid, to maintain the set of disks that intersect it. This last problem, for each cell, reduces to the maintenance of a collection of upper envelopes of unit disks. Although the data structure is not very complicated, programming it would be quite challenging. 

For the unweighted case, we provide a \emph{simple} algorithm that in ${\mathcal O}(n\log n)$ time computes a shortest path tree in $G(P)$ from a given source. Our algorithm is implementable and considerably simpler than the data structure discussed in the previous paragraph or the algorithm of Roditty and Segal. 
For the weighted case, we show how to compute a shortest path tree in ${\mathcal O}(n^{1+\varepsilon})$ time. (Here, $\varepsilon$ denotes an arbitrary positive constant that we can choose and affects the constants hidden in the ${\mathcal O}$-notation.)
This is a signif\mbox{}icant improvement over the result of Roditty and Segal. In this  case we use a simple modif\mbox{}ication of Dijkstra's algorithm combined with a data structure to dynamically maintain a bichromatic closest pair under an Euclidean weighted distance. 
    
Gao and Zhang~\cite{gz-05} showed that the metric induced by a unit disk graph admits a compact well separated pair decomposition, extending the celebrated result of Callahan and Kosaraju~\cite{ck-95} for Euclidean spaces. For making use of the well separated pair decomposition, Gao and Zhang~\cite{gz-05} obtain a $(1+\varepsilon)$-approximation to shortest path distance in unit disk graphs in ${\mathcal O}(n\log n)$ time. Here we provide exact computation within comparable bounds.

Chan and Efrat~\cite{ce-01} consider a graph def\mbox{}ined on a point set but with more general weights in the edges. Namely, it is assumed that there is a function $\ell\colon \RR^2\times \RR^2 \rightarrow \RR_+$ such that the edge $pp'$ gets weight $\ell(p,p')$. Moreover, it is assumed that the function $\ell(p,p')$ is increasing with $\| p-p'\|$. When $\ell(p,p')=\| p-p'\|^2 f(\| p-p'\|)$ for a monotone increasing function $f$, then a shortest path can be computed in ${\mathcal O}(n \log n)$ time. Otherwise, if $\ell$ has constant size description, a shortest path can be computed in roughly ${\mathcal O}(n^{4/3})$ time.

There has been a vast amount of work on algorithmic problems for unit disk and a  review is beyond our possibilities. In the seminal paper of Clark, Colbourn and Johnson~\cite{ccj-90} it was shown that several $\mathcal{NP}$-hard optimization problems remain hard for unit disk graphs, although they showed the notable exception that maximum clique is solvable in polynomial time. Hochbaum and Maass~\cite{hm-85} gave polynomial time approximation schemes for f\mbox{}inding a largest independent set problems using the so-called shifting technique and there have been several developments since.

Shortest path trees can be computed for unit disk graphs in polynomial time. One can just construct $G(P)$ explicitly and run a standard algorithm for shortest paths. The main objective here is to obtain a faster algorithm that avoids the explicit construction of $G(P)$ and exploits the geometry of $P$. There are several problems that can be solved in polynomial time, but faster algorithms are known for geometric settings. A classical example is the computation of the minimum spanning tree of a set of points in the Euclidean plane. Using the Delaunay triangulation, the number of relevant edges is reduced from quadratic to linear.
For more advanced examples see Vaidya~\cite{Vaidya-89}, Efrat, Itai and Katz~\cite{eik-01}, Eppstein~\cite{eppstein-09}, or Agarwal, Overmars and Sharir~\cite{aos-06}.

\paragraph{Organization} In Section~\ref{sec:unweighted} we consider the unweighted case and in Section~\ref{sec:weighted} we consider the weighted case. 
We conclude listing some open problems.

\section{Unweighted shortest paths}
\label{sec:unweighted}

In this section we consider the unweighted version of $G(P)$ and compute a shortest path tree from a given point $s\in P$. Pseudocode for the eventual algorithm is provided in Figure~\ref{fig:BFS}. Before moving into the details, we provide the main ideas employed in the algorithm.

As it is usually done for shortest path algorithms we use tables $\dist[\cdot]$ and $\pi[\cdot]$ indexed by the points of $P$ to record, for each point $p\in P$, the distance $d(s,p)$ and the ancestor of $p$ in a shortest $(s,p)$-path. We start by computing the Delaunay triangulation $DT(P)$ of $P$. We then proceed in rounds for increasing values of $i$, where at round $i$ we f\mbox{}ind the set $W_i$ of points at distance exactly $i$ in $G(P)$ from the source $s$. We start with $W_0=\{ s\}$. At round $i$, we use $DT(P)$ to grow a neighbourhood around the points of $W_{i-1}$ that contains $W_{i}$. More precisely, we consider the points adjacent to $W_{i-1}$ in $DT(P)$ as candidate points for $W_{i}$. For each candidate point that is found to lie in $W_{i}$, we also take its adjacent vertices in $DT(P)$ as new candidates to be included in $W_{i}$. For checking whether a candidate point $p$ lies in $W_{i}$ we use a data structure to f\mbox{}ind the nearest neighbour of $p$ in $W_{i-1}$, denoted by $\func{NN}(W_{i-1},p)$. Such data structure is just a point location data structure in the Voronoi diagram of $W_{i-1}$. Similarly, the shortest path tree is constructed by connecting each point of $W_i$ to its nearest neighbour in $W_{i-1}$. See Figure~\ref{fig:BFS} for the eventual algorithm \proc{UnweightedShortestPath}. In Figure~\ref{fig:exampleBFS} we show an example of what edges of the shortest path tree are computed in one iteration of the main loop.

\begin{figure}[htb]
\centering
\ovalbox{
\begin{minipage}{.9\hsize}
\begin{codebox}
    \Procname{$\proc{UnweightedShortestPath}(P,s)$}
    \li \For$p\in P$ \Do
    	\li $\dist[p] \gets \infty$
    	\li $\pi[p] \gets \const{nil}$\End
    \li $\dist[s] \gets  0$
    \li build the Delaunay triangulation $DT(P)$
    \li $W_0 \gets \{ s\}$
    \li $i\gets1$
    \li \While $W_{i-1}\neq\emptyset$ \Do
    	\li build data structure for nearest neighbour queries in $W_{i-1}$
    	\li $Q \gets W_{i-1}$ \>\>\>\Comment candidate points
        \li $W_{i}\gets\emptyset$
    	\li \While $Q\neq\emptyset$\Do
    		\li $q$ an arbitrary point of $Q$
            \li remove $q$ from $Q$
            \li \For $qp$ edge in $DT(P)$ \Do
            	\li $w \gets \func{NN}(W_{i-1},p)$
            	\li \If $\dist[p]=\infty$ and $\| p-w\|\leq 1$ \Then
                	\li $\dist[p]\gets i$
                    \li $\pi[p]\gets w$
                    \li add $p$ to $Q$
                    \li add $p$ to $W_{i}$
                    \End
            	\End
            \End
        \li $i\gets i+1$
    \End
    \li \Return $\dist[\cdot]$ and $\pi[\cdot]$
\end{codebox}
\end{minipage}}
\caption{Algorithm to compute a shortest path tree in the unweighted case.}
\label{fig:BFS}
\end{figure}

\begin{figure}[htb]
\centering
	\includegraphics[scale=.8,page=2]{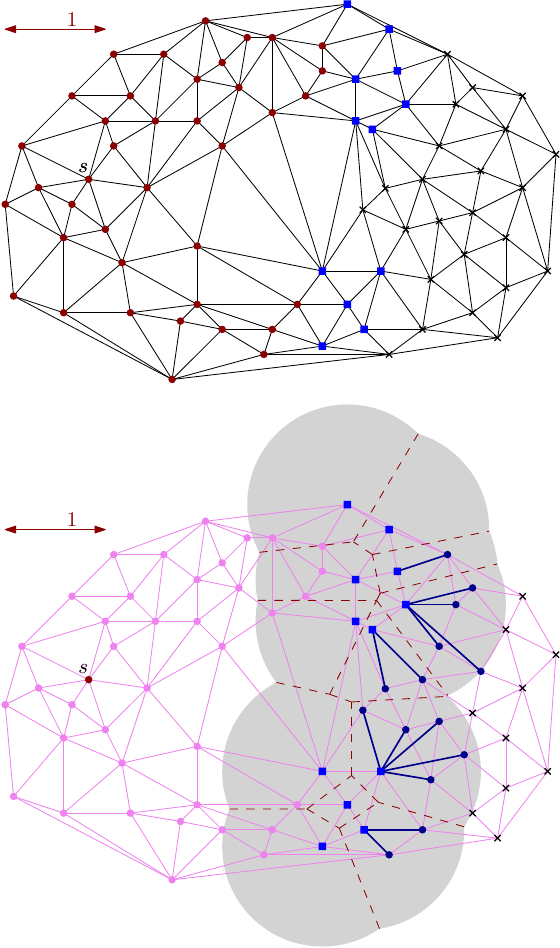}
	\caption{Top: A point set with its Delaunay triangulation. 
        The source is marked as $s$. Points $p$ such that 
        $d(s,p)\le 3$ are marked with red dots. Points 
        from $W_4$ are marked with blue boxes.
        Bottom: The new edges added to the tree at iteration $5$
        and the new vertices are shown. The light grey region is
        $\bigcup_{p\in W_4} D(p,1)$, where the Voronoi diagram of $W_4$ 
        is superimposed.}
\label{fig:exampleBFS}
\end{figure}

We would like to emphasize a careful point that we employ to achieve the running time $\mathcal{O}(n \log n)$. For any point $p$, let $D(p,1)$ denote the disk of radius $1$ centered at $p$.
In lines 16 and 17 of the algorithm, we check whether $p$ is at distance at most $1$ from \emph{some} point in $W_{i-1}$, namely its nearest neighbour in $W_{i-1}$. Checking whether $p$ is at distance at most $1$ from $\pi(q)$ (or $q$ when $q\in W_{i-1}$) would lead to a potentially larger running time. Thus, we do not grow each disk $D(w,1)$ independently for each $w\in W_{i-1}$, but we grow the whole region $\bigcup_{w\in W_{i-1}} D(w,1)$ at once. Growing each disk $D(w,1)$ separately would force as to check the same edge $qp$ of $DT(P)$ several times, once for each $w\in W_{i-1}$ such that $q\in D(w,1)$.

\begin{figure}[htb]
\centering
	\includegraphics[]{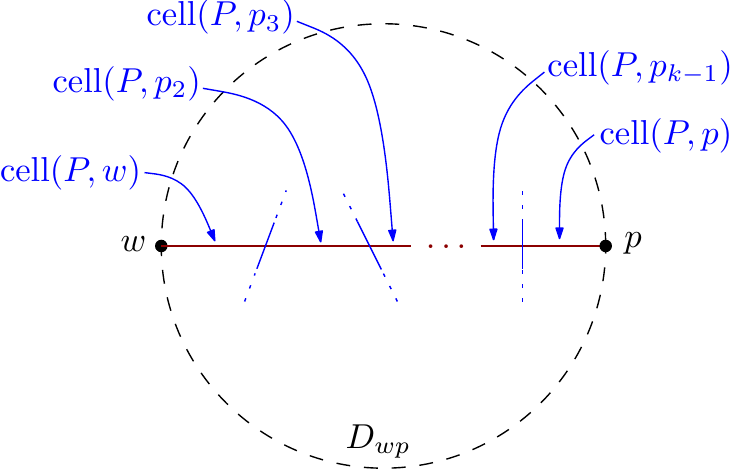}
	\caption{Proof of Lemma~\ref{le:path}.}
\label{fig:path}
\end{figure}

\begin{lemma}
\label{le:path}
	Let $p$ be a point from $P\setminus \{ s \}$ such that $d(s,p)<\infty$.
    There exists a point $w$ in $P$ and a path $\pi$ in $DT(P)\cap G(P)$ from
    $w$ to $p$ such that $d(s,w)+1=d(s,p)$ and
    each internal vertex $p_j$ of $\pi$ satisf\mbox{}ies $d(s,p_j)= d(s,p)$.
\end{lemma}

\begin{proof}
	Let us set $i=d(s,p)$
    Let $w$ be the point with $d(s,w)=i-1$ that is closest to $p$ in Euclidean
    distance. It must be that $\| w-p\|\le 1$ because $d(s,p)<\infty$.
	Let $D_{wp}$ be the disk with diameter $wp$. 
    
    For simplicity, let us assume that the segment $wp$ 
    does not go through any vertex of the Voronoi diagram of $P$.
    (In the degenerate case where $wp$ goes through a vertex of the Voronoi diagram, 
	we can replace $p$ by a 
    point $p'$ arbitrarily close to $p$.)
    Consider the sequence of Voronoi cells 
    $\cell(p_1,P),\,\dots,\,\cell(p_k,P)$ intersected by the segment $wp$, 
    as we walk from $w$ to $p$. See Figure~\ref{fig:path}.
    Clearly $w=p_1$ and $p=p_k$. For each $1\le j< k$, the edge $p_j p_{j+1}$ is in
    $DT(P)$ because $\cell(p_j,P)$ and $\cell(p_{j+1},P)$ 
    are adjacent along some point of $wp$. 
    Therefore the path $\pi=p_1 p_2\dots p_k$ is contained
    in $DT(P)$ and connects $w$ to $p$. 
    For any index $j$ with $1< j< k$, let $a_j$ be any point in 
    $wp \cap \cell(p_j,P)$.
    Since $\| a_j p_j\| \le \min \{ \| a_jw\| , \| a_jp \| \}$, 
    the point $p_j$ is contained in $D_{wp}$. 
    Therefore the whole path $\pi$ is contained in $D_{wp}$ and,
    since $D_{wp}$ has diameter at most one,
    each edge of $\pi$ is also in $G(P)$. We conclude that
    $\pi$ is a path in $DT(P)\cap G(P)$.
    
    Consider any point $p_j$ of $\pi$, which is thus contained in $D_{wp}$.
    Because $\| w - p_j\|\le \|w-p\|\le 1$, we have $d(s,p_j)\le d(s,w)+1 = i$. 
    Because $\| p_j -p\|\le \|w-p\|\le 1$, we have $d(s,p_j)\ge d(s,p)-1 =i-1$.
    However, the choice of $w$ as closest to $p$ implies that 
    $d(s,p_j)\not= i-1$ because $\| p_j -p\|< \|w-p\|$.
    Therefore $d(s,p_j)=i$. 
    We conclude that all internal vertices $p_j$ of $\pi$ satisfy $d(s,p_j)=i$.
\end{proof}

\begin{lemma}
\label{le:correct}
	At the end of algorithm \proc{UnweightedShortestPath}$(P,S)$ it holds 
    \[ \forall i\in \mathbb{N}\cup \{ 0\}: ~~~ W_i ~=~ \{ p\in P \mid d(s,p)=i\}. \]
    Moreover, for each point $p\in P\setminus\{ s \}$, 
    it holds that $\dist[p]=d(s,p)$ and, if $d(s,p)<\infty$, 
    there is a shortest path in $G(P)$ from $s$ to $p$ that uses 
    $\pi[p]p$ as last edge.    
\end{lemma}

\begin{proof}
	We prove the statement by induction on $i$. 
    $W_0=\{ s\}$ is set in line 6 and never changed. Thus the statement holds
    for $i=0$.
    
    Before considering the inductive step, note that the sets
    $W_0,\,W_1,\,\dots$ are pairwise disjoint. Indeed, a point $p$ is added to 
    some $W_i$ (line 21) at the same time that we set $\dist[p]=i$ (line 18).
    After setting $\dist[p]$, the test in line 17 is always false and
    $p$ is not added to any other set $W_j$.
    
    Consider any value $i>0$. By induction we have that
    \[ W_{i-1} ~=~ \{ p\in P \mid d(s,p)=i-1\}. \]
    In the algorithm we add points to $W_i$ only in line 21. If a point $p$
    is added to $W_{i}$, then $\| p-w\|\le 1$ for some $w\in W_{i-1}$
    because of the test in line 17. Therefore any point $p$ added to
    $W_{i}$ satisf\mbox{}ies $d(s,p)\le i$. Since $p\notin W_{i-1}$, the disjointness
    of the sets $W_0,W_1,\dots,$ implies that $d(s,p)= i$.
    We conclude that
    \[
    	 W_i ~\subseteq ~ \{ p\in P \mid d(s,p)=i\}.
    \]

    For the reverse containment, let $p$ be any point such that
    $d(s,p)=i$. We have to show that $p$ is added to $W_i$ by the algorithm.
    Consider the point $w$ and the path $\pi=p_1\dots p_k$ 
    guaranteed by Lemma~\ref{le:path}.    
    By the induction hypothesis, $w=p_1 \in W_{i-1}$
    and thus is added to $Q$ in line 10.
    At some moment the edge $p_1p_2$ is considered in line 15
    and the point $p_2$ is added to $W_i$ and $Q$.
    An inductive argument thus shows that all the points 
    $p_{3},\dots, p_k$ are added to $W_i$ and $Q$ 
    (possibly in a dif\mbox{}ferent order).
    It follows that $p_k=p$ is added to $W_i$ and thus    
    \[
    	 W_i ~=~ \{ p\in P \mid d(s,p)=i\}.
    \]    
    Since a point $p$ is added to $W_i$ at the same time that
    $\dist[p]=i$ is set, it follows that $\dist[p]=i=d(s,p)$.
    Since $\pi[p]\in W_{i-1}$ and $\| p-\pi[p]\| \le 1$ (lines 16, 17 and 19),
    there is a shortest path in $G(P)$ from $s$ to $p$ that uses 
    an $(i-1)$-edge path from $s$ to $\pi[p]$, by induction, followed by 
    the edge $\pi[p]p$.
\end{proof}

\begin{lemma}
\label{le:time}
	The algorithm \proc{UnweightedShortestPath}$(P,S)$ takes
    $\mathcal{O}(n \log n)$ time, where $n$ is the size of $P$.
\end{lemma}

\begin{proof}
	For each point $q\in P$, let $\deg_{DT(P)}(q)$ denote the degree of $q$ in 
	the Delaunay triangulation $DT(P)$.
	The main observations used in the proof are the following:
    each point of $P$ is added to $Q$ at most once in line 10 and
    once in line 20, the execution
    of lines 13--21 for a point $q$ takes time $\mathcal{O}(\deg_{DT(P)}(q)\log n)$, 
    the sum of the degrees of the points in $DT(P)$ is $\mathcal{O}(n)$,
    and in line 9 we spend time $\mathcal{O}(n \log n)$ overall iterations together.
    We next provide the details.
    
	The Delaunay triangulation of $n$ points can be computed in $\mathcal{O}(n \log n)$ time.
    Thus the initialization in lines 1--7 takes $\mathcal{O}(n \log n)$ time.
    It remains to argue that the loop in lines 8--22 takes time $\mathcal{O}(n\log n)$.
    
	An execution of the lines 9--11 takes time 
    $\mathcal{O}(|W_{i-1}|\log |W_{i-1}|)=\mathcal{O}(|W_{i-1}|\log n)$. Each
    subsequent nearest neighbour query takes $\mathcal{O}(\log n)$ time.
    
    Each execution of the lines 16--21 takes time ${\mathcal O}(\log n)$,
    where the most demanding step is the query made in line 16.
    Each execution of the lines 13--21 takes time 
    ${\mathcal O}(\deg_{DT(P)}(q)\cdot \log n)$ because the lines 16--21
    are executed $\deg_{DT(P)}(q)$ times.
    
	Consider one execution of the lines 9--22 of the algorithm.
    Points are added to $Q$ in lines 10 and 20. In the latter case, a point $p$
    is added to $Q$ if and only if it is added to $W_{i}$ (line 21). 
    It follows that a point is added to $Q$ if and only if it belongs 
    to $W_{i-1}\cup W_{i}$. Moreover, each point of $W_{i-1}\cup W_{i}$ is added
    exactly once to $Q$: each point $p$ that is added to $Q$ has 
    $\dist[p]\le i<\infty$
    and will never be added again because of the test in line 17.
	It follows that the loop in lines 12--22 takes time 
	\[
    	\sum_{q\in W_{i-1}\cup W_{i}} \mathcal{O}(\deg_{DT(P)}(q)\cdot \log n).
	\]
    Therefore we can bound the time spent in the the loop of lines 8--22 by
    \begin{equation}
    \label{eq:bound1}
    	\sum_i \mathcal{O}\left( |W_i|\log n + 
        \sum_{q\in W_{i-1}\cup W_{i}} \left(\deg_{DT(P)}(q)\cdot \log n\right)\right).
    \end{equation}    
    Using that the sets $W_0,\,W_1,\,\dots$ are pairwise 
    disjoint (Lemma~\ref{le:correct}) with $\sum_i |W_i|\le n$ and
    \[
    	\sum_{q\in P} \deg_{DT(P)}(q)~=~ 2\cdot |E(DT(P))| = \mathcal{O}(n),
    \]
    the bound in \eqref{eq:bound1} becomes $\mathcal{O}(n\log n)$.    
\end{proof}

\begin{thm}
  Let $P$ be a set of $n$ points in the plane and let $s$ be a point from $P$. 
  In time ${\mathcal O}(n \log n)$ we can compute a shortest path tree from $s$
  in the unweighted graph $G(P)$.
\end{thm}
\begin{proof}
  Consider the algorithm \proc{UnweightedShortestPath}$(P,S)$ given in 
  Figure~\ref{fig:BFS}. Because of Lemma~\ref{le:time} it takes time ${\mathcal O}(n\log n)$.
  Because of Lemma~\ref{le:correct}, the table $\pi[\cdot]$ correctly describes
  a shortest path tree from $s$ in $G(P)$ and $\dist[\cdot]$ correctly describes shortest path distances in $G(P)$.
\end{proof}

\section{Weighted shortest paths}
\label{sec:weighted}

In this section we consider the \texttt{SSSP} problem on the weighted version of $G(P)$: points $p$ and $q$ have an edge between them if\mbox{}f $\|p-q\|\leq1$ and the weight of that edge is $\|p-q\|$. Our algorithm uses a dynamic data structure for bichromatic closest pairs. We f\mbox{}irst review the precise data structure that we will employ. We then describe the algorithm and discuss its properties.

\subsection{Bichromatic closest pair}

In the bichromatic closest pair problem, we are given a set of red points and
a set of blue points in a metric space, and we have to find the pair
of points, one of each colour, that are closest. Many versions and generalizations 
of this basic problem have been studied. Here, we are interested 
in a dynamic version with a functional reminiscent of distances.

Let $P$ be a set of $n$ points in the plane and let each point $p\in P$ have a weight $w_p\geq0$. We call a function $\delta\colon\mathbb{R}^2\times P\longrightarrow\mathbb{R}_+$ a \emph{(additive) weighted Euclidean metric}, if it is of the form
$$\delta(q,p)=w_p+\|q-p\|,$$
where $\|\cdot\|$ denotes the Euclidean distance.

Let $\varepsilon>0$ denote an arbitrary constant. Agarwal, Efrat and Sharir \cite{aes-99} showed that for any $P$ and $\delta$ as above, $P$ can be preprocessed in $\mathcal{O}(n^{1+\varepsilon})$ time into a data structure of size $\mathcal{O}(n^{1+\varepsilon})$ so that points can be inserted into or deleted from $P$ in $\mathcal{O}(n^\varepsilon)$ amortized time per update, and a nearest-neighbour query can be answered in $\mathcal{O}(\log n)$ time. Eppstein \cite{eppstein-95} had already shown that if such a dynamic data structure existed, then a bichromatic closest pair (BCP) under $\delta$ of red and blue points in the plane could be maintained, adding only a polylogarithmic factor to the update time. Combining these two results gives

\begin{thm}[Agarwal, Efrat, Sharir \cite{aes-99}]\label{thm:aes-bcp}
Let $R$ and $B$ be two sets of points in the plane with a total of $n$ points. We can store $R\cup B$ in a dynamic data structure of size $\mathcal{O}(n^{1+\varepsilon})$ that maintains a bichromatic closest pair in $R\times B$, under any weighted Euclidean metric, in $\mathcal{O}(n^\varepsilon)$ amortized time per insertion or deletion.\qed
\end{thm}

\subsection{Algorithm}

We will use a variant of Dijkstra's algorithm. As before, we maintain tables $\dist[\cdot]$ and $\pi[\cdot]$ containing distances from the source and parents of points in the shortest path tree. As in Dijkstra's algorithm we will maintain a set $S$ (containing the source $s$) of points for which the correct distance from $s$ has already been computed, and a set $P\setminus S$ of points for which the distance has yet to be computed. For the points of $S$, $\dist[\cdot]$ stores the true distance from the source.

In our approach we split $S$ into sets $B$ and $D$, called ``blue'' and ``dead'' points, respectively. We call the points in $R=P\setminus S$ ``red'' points. The reason for the introduction of the ``dead'' points $D$ is that, as it will be proved, during the entire algorithm it is not possible there is an edge of $G(P)$ between a point in $D$ and a point in $R$. Thus, the points of $D$ are not relevant to f\mbox{}ind the last edge in a shortest path to points of $R$.

We store $R\cup B$ in the dynamic data structure from Theorem~\ref{thm:aes-bcp} that maintains the bichromatic closest pair (BCP) in $R\times B$ under the weighted Euclidean metric
$$\delta(r,b):=\dist[b]+\|r-b\|.$$
At each iteration of the main \kw{while} loop, we query the data structure for a BCP pair $(b^\ast,r^\ast)$. If $b^\ast r^\ast$ is not an edge in our underlying graph $G(P)$, meaning $\|b^\ast-r^\ast\|>1$, then $b^\ast$ will never be the last vertex to any point in $R$, and therefore we will move it from $B$ to $D$. If $b^\ast r^\ast$ is an edge of $G(P)$, then, as it happens with Dijkstra's algorithm, we have completed a shortest path to $r^\ast$. The algorithm is given in Figure~\ref{fig:algo-dijkstra}. Figure~\ref{fig:d-b-r} shows sets $D$, $B$, and $R$ in the middle of a run of the algorithm.

\begin{figure}[h!]
\begin{center}
\ovalbox{
\begin{minipage}{\hsize}
\begin{codebox}
\Procname{$\proc{WeightedShortestPaths}(P,s)$}
\li \For$p\in P$
\li \Do$\dist[p]\gets\infty$
\li $\pi[p]\gets\const{nil}$\End
\li $\dist[s]\gets0$
\li $B\gets\{s\}$
\li $D\gets\emptyset$
\li $R\gets P\setminus\{s\}$
\li store $R\cup B$ in the BCP dynamic DS of Theorem~\ref{thm:aes-bcp} wrt $\delta(r,b)$
\li \While$R\neq\emptyset$
\li \Do\If$B=\emptyset$
\li \Then\Return $\dist[\cdot]$ and $\pi[\cdot]$  \>\>\>\>\>\>\Comment$G(P)$ is not connected
\li \Else$(b^\ast,r^\ast)\gets\func{BCP}(B,R)$\End
\li \If$\|b^\ast-r^\ast\|>1$
\li \Then$\func{delete}(B,b^\ast)$
\li $D\gets D\cup\{b^\ast\}$
\li \Else$\dist[r^\ast]\gets\dist[b^\ast]+\|b^\ast-r^\ast\|$
\li $\pi[r^\ast]\gets b^\ast$
\li $\func{delete}(R,r^\ast)$
\li $\func{insert}(B,r^\ast)$\End\End
\li \Return $\dist[\cdot]$ and $\pi[\cdot]$
\end{codebox}
\end{minipage}
}
\end{center}
\caption{Algorithm for \texttt{SSSP} in the weighted case.}
\label{fig:algo-dijkstra}
\end{figure}

\begin{figure}[h!]
\begin{center}
\includegraphics{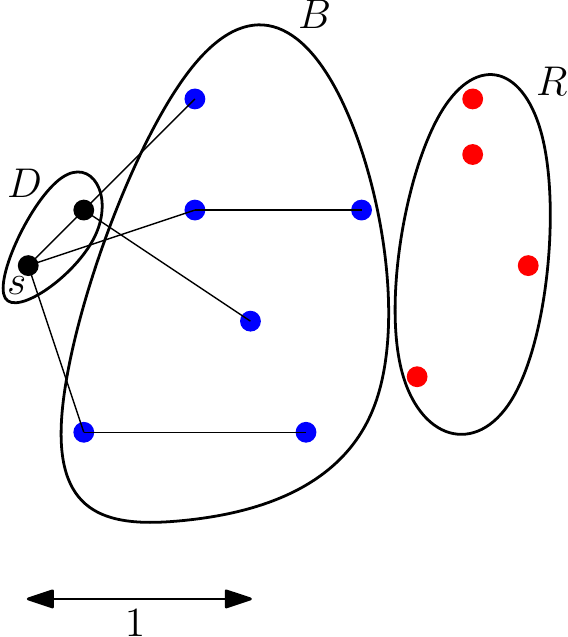}
\end{center}
\caption{Sets $D$, $B$, and $R$ after a couple of iterations of the \kw{while} loop in \proc{WeightedShortestPaths}.}
\label{fig:d-b-r}
\end{figure}

Let us explain the actual bottleneck of our approach to reduce the time from ${\mathcal O}(n^{1+\varepsilon})$ to ${\mathcal O}(n\polylog n)$. The inner workings of the data structure of Theorem~\ref{thm:aes-bcp} is based on two dynamic data structures. One of them has to compute $\min_{b\in B }\delta(r_0,b)$ for a given $r_0\in R$. The other has to compute $\min_{r\in R}\delta(r,b_0)$ for a given $b_0\in B$. For the latter data structure we could use the dynamic nearest neighbour data structure by Chan \cite{chan-10}, yielding polylogarithmic update and query times. However, for the former we need a dynamic weighted Voronoi diagram, and for this we only have the data structure developed by Agarwal, Efrat and Sharir~\cite{aes-99}. A dynamic data structure for dynamic weighted Voronoi diagrams with updates and queries in polylogarithmic time readily would lead to ${\mathcal O}(n\polylog n)$. 

\subsection{Correctness and Complexity}

Note that in the algorithm a point can only go from red to blue and from blue to dead. Dead points stay dead.
We f\mbox{}irst prove two minor properties.
 
\begin{lemma}\label{le:dead}
Once a point $b^\ast$ is moved from $B$ to $D$, it no longer has any edges to points in $R$.
\end{lemma}

\begin{proof}
The move of $b^\ast$ from $B$ to $D$ is a consequence of two facts: i) $(b^\ast,r^\ast)$ is a BCP in $B\times R$ --- achieving the minimum of the expression
$$\min_{r\in R,\,b\in B}\delta(r,b) ~=~ \min_{r\in R,\,b\in B}\{\dist[b]+\|r-b\|\},$$
and ii) $\|b^\ast-r^\ast\|>1$. Therefore, $\forall r\in R\colon\|b^\ast-r\|>1$.
\end{proof}

\begin{lemma}\label{le:connected}
$G(P)$ is not connected if and only if there is a moment in the algorithm when it holds that $R\neq\emptyset$ and $B=\emptyset$.
\end{lemma}
\begin{proof}
($\Rightarrow$): If $G(P)$ is not connected, there is a point $r$ that begins in $R$ and is not reachable from $s$. It never leaves $R$, so $R$ stays nonempty throughout. $B$ gets emptied to $D$ once the data structure starts returning only BCPs that do not form an edge in $G(P)$.

\noindent($\Leftarrow$): By Lemma~\ref{le:dead} the dead points do not have any edges to red points. If there are no blue points then $G(P)$ is not connected.
\end{proof}

\begin{lemma}\label{le:dijkstra-correct}
The algorithm \proc{WeightedShortestPaths} correctly computes the shortest distances from the source and the parents of points in a SSSP tree.
\end{lemma}

\begin{proof}
As in Dijkstra's algorithm, we f\mbox{}ind the vertex $r^\ast$ minimizing the expression $\func{dist}[b]+\|b-r\|$ over all vertices $b\in B\cup D$ and $r\in R$ with $\|b-r\|\le 1$, and update the information of $r^\ast$ accordingly. Thus the correctness follows from the correctness of Dijkstra's algorithm. 
\end{proof}

\begin{lemma}\label{le:dijkstra-time}
The algorithm \proc{WeightedShortestPaths} runs in $\mathcal{O}(n^{1+\varepsilon})$ time and space, for an arbitrary constant $\varepsilon>0$.
\end{lemma}

\begin{proof}
The outer while loop runs at most $2n-2$ times, as in each iteration either a blue point is deleted and placed among the dead, or a red point becomes blue. If $G(P)$ is not connected, the loop terminates even earlier by Lemma~\ref{le:connected}. In each iteration, either we f\mbox{}inish because $B=\emptyset$, or we spend ${\mathcal O}(1)$ time plus the time to make ${\mathcal O}(1)$ operations in the BCP dynamic DS. Since by Theorem~\ref{thm:aes-bcp} each operation in the BCP dynamic DS takes $\mathcal{O}(n^\varepsilon)$ amortized time, the result follows.
\end{proof}

\begin{thm}
Let $P$ be a set of $n$ points in the plane, $s\in P$, and $\varepsilon>0$ an arbitrary constant. The algorithm $\proc{WeightedShortestPaths}(P,s)$ returns the correct distances from the source in the graph $G(P)$ in $\mathcal{O}(n^{1+\varepsilon})$ time.
\end{thm}

\begin{proof}
By Lemmata~\ref{le:dijkstra-correct} and \ref{le:dijkstra-time}.
\end{proof}

\section{Conclusions}
We have given algorithms to compute shortest paths in unit disk graphs in near-linear time. For the unweighted case it is easy to show that our algorithm is asymptotically optimal in the algebraic decision tree. A simple reduction from the problem of f\mbox{}inding the maximum gap in a set of numbers shows that deciding if $G(P)$ is connected requires $\Omega(n\log n)$ time. As discussed in the text, a better data structure to dynamically maintain the bichromatic closest pair would readily imply an improvement in our time bounds for the weighted case.

A generalization of the graph $G(P)$ is the graph $G_{\le t}(P)$, where two points are connected whenever their distance is at most $t$. Thus $G(P)$ is $G_{\le 1}(P)$. Two natural extensions of our results come to our mind. 
\begin{itemize}
\item Can we compute ef\mbox{}f\mbox{}iciently a compact representation of the distances in all the graphs $G_{\le t}(P)$? 
\item Can we f\mbox{}ind,
for a given $u,v,k$, the minimum $t$ such that in $G_{\le t}(P)$ the distance from $u$ to $v$ is at most $k$? This problem can be solved in roughly $O(n^{4/3})$ time using~\cite{selecting,expander} to guide a binary search in the interpoint distances. Can it be solved in near-linear time?
\end{itemize}

\section*{Acknowledgments}
We would like to thank Timothy Chan, Alon Efrat, and David Eppstein for several useful comments. In particular, we are indebted to Timothy Chan for pointing out the work of Roditty and Segal~\cite{rs-11} and to Alon Efrat for explaining the alternative algorithm for the unweighted case discussed in the introduction.
\bibliographystyle{abbrv}
\bibliography{biblio}

\end{document}